\documentclass{article}

\usepackage[utf8]{inputenc}
\usepackage[left=3cm, right=3cm, top=2cm, bottom=2cm]{geometry}

\usepackage[linesnumbered,boxed]{algorithm2e}
\usepackage[noend]{algpseudocode}

\usepackage{tikz}

\usepackage{amsmath}
\usepackage{amssymb}
\usepackage{amsthm}
\usepackage{authblk}
\usepackage{comment}
\usepackage{color}
\usepackage{dsfont}
\usepackage{hyperref}
\usepackage{mdframed}

\newtheorem{definition}{Definition}
\newtheorem{theorem}{Theorem}
\newtheorem{corollary}[theorem]{Corollary}
\newtheorem{lemma}[theorem]{Lemma}

\newtheorem{claim}[theorem]{Claim}

\newcommand{\eps}{\varepsilon}

\newcommand{\rr}{\mathbb{R}}

\newcommand{\cutnorm}[1]{\left\| #1 \right\|_{\square}}

\newcommand{\norm}[2]{\left\| #2 \right\|_{#1}}
\newcommand{\ip}[3]{\left\langle #2, #3 \right\rangle_{#1}}

\newcommand{\onevec}{\mathds{1}}

\newcommand{\frobnorm}[1]{\left\| #1 \right\|_{F}}

\usepackage{booktabs}

\DeclareMathOperator{\vol}{vol}

\DeclareMathOperator{\tower}{tower}
\DeclareMathOperator{\poly}{poly}

\DeclareMathOperator{\tr}{Tr}

\DeclareMathOperator{\spann}{span}

\newcommand{\Szemeredi}{Szemer\'{e}di}

\newcommand{\Lovasz}{Lov{\' a}sz}

\title{A Unified View of Graph Regularity via Matrix Decompositions\footnote{Supported in part by NSF awards CCF-1717349,  DMS-183932 and CCF-1909756.  bodwin@umich.edu, vempala@cc.gatech.edu.}}
\author[1]{Greg Bodwin}
\author[2]{Santosh Vempala}
\affil[1]{University of Michigan}
\affil[2]{Georgia Tech}

\date{}

\begin{document}

\maketitle

\thispagestyle{empty}

\begin{abstract}
We give a unified proof of sparse algorithmic weak and \Szemeredi{} regularity lemmas for several well-studied graph classes where only sparse weak regularity lemmas were previously known.
These include core-dense graphs, low threshold rank graphs, and (a version of) $L^p$ upper regular graphs.
More precisely, we define \emph{cut pseudorandom graphs}, we prove our sparse regularity lemmas for these graphs, and then we show that cut pseudorandomness captures all of the above graph classes as special cases.

The core of our approach is an abstracted matrix decomposition, which can be computed using a simple algorithm by Charikar [AAC0 '00].
Using work of Oveis Gharan and Trevisan [TOC '15], it also implies new PTASes for MAX-CUT, MAX-BISECTION, MIN-BISECTION for a significantly expanded class of input graphs.
(It is NP Hard to get PTASes for these graphs in general.)
\end{abstract}

\pagebreak

\setcounter{page}{1}

\begin{center}
\begin{figure}[h!]
\begin{tikzpicture} [scale=0.8]

\draw [gray!10, fill=gray!10] (-8, 6) rectangle (2.5, -3.5);
\node at (-2.5, 5.25) {\bf Graph Classes};
\node at (7.5, 5.25) {\bf Applications};

\draw [thick, <-] (-5, 2.25) -- (-5, 1.75);

\draw [thick, ->] (-3, -1) -- (-2, 0);
\draw [thick, ->] (-3, -1) -- (-2, -2);
\draw [thick, ->] (-3, 3) -- (-2, 4);
\draw [thick, ->] (-3, -1) -- (-2, 4);

\draw [thick, ->] (0, 3.25) -- (0, 2.75);

\draw [thick, ->] (2, 4) -- (3, 2);
\draw [thick, ->] (2, 2) -- (3, 0);
\draw [thick, ->] (2, 0) -- (3, 0);
\draw [thick, ->] (2, -2) -- (3, 0);

\draw [thick, ->] (7, 0) -- (8, 0);
\draw [thick, ->] (7, 0) -- (8, 2);


\draw (-5, -1) ellipse (2 and 0.75) node {Dense};
\draw (-5, 1) ellipse (2 and 0.75) node {$(n, d, \lambda)$ Graphs};
\draw (-5, 3) ellipse (2 and 0.75) node {Jumbled};

\draw [align=center] (0, 4) ellipse (2 and 0.75) node {$L_{\infty}$ Upper Reg};
\draw [align=center] (0, 2) ellipse (2 and 0.75) node {$L_{p}$ Upper Reg};
\draw (0, 0) ellipse (2 and 0.75) node {Core Dense};
\draw [align=center] (0, -2) ellipse (2 and 0.75) node {Low Thresh Rank};

\draw (5, 2) ellipse (2 and 0.75) node {Szem.\ Regularity};
\draw (5, 0) ellipse (2 and 0.75) node {Cut Approx};

\draw (10, 2) ellipse (2 and 0.75) node {Weak Regularity};
\draw (10, 0) ellipse (2 and 0.75) node {Approx Algs};

\end{tikzpicture}

\hrulefill

~

\begin{tikzpicture} [scale=0.8]
\draw [gray!10, fill=gray!10] (-8, 6) rectangle (7.5, -3.5);
\node at (0, 5.25) {\bf Graph Classes};
\node at (10, 5.25) {\bf Applications};

\draw [thick, <-] (-5, 2.25) -- (-5, 1.75);

\draw [thick, ->] (-3, -1) -- (-2, 0);
\draw [thick, ->] (-3, -1) -- (-2, -2);
\draw [thick, ->] (-3, 3) -- (-2, 4);
\draw [thick, ->] (-3, -1) -- (-2, 4);

\draw [thick, ->] (0, 3.25) -- (0, 2.75);

\draw [ultra thick, ->, blue] (2, 2) -- (3, 1);
\draw [ultra thick, ->, blue] (2, 0) -- (3, 1);
\draw [ultra thick, ->, blue] (2, -2) -- (3, 1);

\draw [blue, ultra thick, ->] (7, 1) -- (8, 0);
\draw [thick, <-] (10, 1.25) -- (10, 0.75);
\draw [ultra thick, ->, blue] (7, 1) -- (8, 4);

\draw [thick, ->] (10, -0.75) -- (10, -1.25);

\draw (-5, -1) ellipse (2 and 0.75) node {Dense};
\draw (-5, 1) ellipse (2 and 0.75) node {$(n, d, \lambda)$ Graphs};
\draw (-5, 3) ellipse (2 and 0.75) node {Jumbled};

\draw [align=center] (0, 4) ellipse (2 and 0.75) node {$L_{\infty}$ Upper Reg};
\draw [align=center] (0, 2) ellipse (2 and 0.75) node {$L_{p}$ Upper Reg};
\draw (0, 0) ellipse (2 and 0.75) node {Core Dense};
\draw [align=center] (0, -2) ellipse (2 and 0.75) node {Low Thresh Rank};

\draw [ultra thick, blue] (5, 1) ellipse (2 and 0.75) node {\color{blue} Cut Pseudorandom};

\draw (10, 4) ellipse (2 and 0.75) node {Szem.\ Regularity};
\draw (10, 2) ellipse (2 and 0.75) node {Weak Regularity};
\draw (10, 0) ellipse (2 and 0.75) node {Cut Approx};
\draw (10, -2) ellipse (2 and 0.75) node {Approx Algs};

\end{tikzpicture}

\caption{\label{fig:results} Some existing implications in the area of graph regularity (top), and how the results in this paper change the picture (bottom).  For clarity we have omitted arrows in these diagrams that are implied transitively by a chain of other arrows, even when these direct implications may have advantages (e.g.\ better quantitative bounds or faster computation).  Citations and explanations are given in the text below.}
\end{figure}
\end{center}

\clearpage


\section{Introduction}

In \emph{graph regularity}, the goal is roughly to partition the nodes of an input graph into a few parts so that the graph looks like a random bipartite graph between most pairs of parts.
The area was pioneered by a classic lemma of \Szemeredi{} \cite{Szemeredi75}, which guarantees the existence of a good partition for any sufficiently dense input graph.
This \Szemeredi{} regularity lemma has many deep algorithmic and combinatorial applications \cite{KSSS00}.
It has also spawned a line of work on regularity lemmas themselves; notable examples include the weak regularity lemma of Frieze and Kannan \cite{FK99}, the strong regularity lemma of Alon, Fischer, Krivelevich and Szegedy \cite{AFKS01}, etc.
These lie at the center of a deep mathematical theory of graph limits \cite{Alon2009,Lovasz12}.

A limitation of these regularity lemmas is that they only comment on sufficiently \emph{dense} input graphs; for example, if a graph has $n$ nodes and $\le n^{1.99}$ edges, then none of these regularity lemmas promise a partition with nontrivial properties.
There are constructions demonstrating that this is unavoidable, and there can be no perfectly analogous regularity lemmas that extend to all sparse graphs \cite{Gowers97, FL14}.
However, it may still be possible to prove regularity lemmas for some graph classes of interest.
This is the basis for \emph{sparse graph regularity}, a research area that investigates which useful classes of sparse graphs admit nontrivial regularity lemmas.
Some notable successes in the area include various sparse regularity lemmas and applications for \emph{jumbled graphs} and \emph{$(n, d, \lambda)$ graphs} (see survey \cite{KS06}), \emph{upper regular graphs} \cite{Kohayakawa97}, \emph{$L_p$ upper regular graphs} \cite{BCCZ14, BCCZ14b}, graphs of \emph{low threshold rank} \cite{GharanT15}, and \emph{core-dense graphs} \cite{KV09}.

The objective of the current paper is to unify and strengthen the linear algebraic foundations of sparse regularity.
We abstract a theorem of Lovasz{} and Szegedy \cite{LS07} into a \emph{matrix decomposition}, and we use this decomposition to define a new class of graphs that we call \emph{cut pseudorandom}.
We then prove:
\begin{itemize}
\item Essentially all classes of graphs mentioned above are cut pseudorandom, and
\item Cut pseudorandom graphs have sparse weak and \Szemeredi{}-type regularity lemmas.  For many of the above graph classes, the sparse \Szemeredi{} regularity lemmas were not previously known.
\item Additionally, there is a simple polynomial-time algorithm that \emph{computes} decompositions satisfying these sparse regularity lemmas for cut pseudorandom graphs.
\end{itemize}

In the following, we overview the matrix decomposition and its consequences more formally.

%
%

\subsection{Cut Decompositions \label{sec:cutdecomp}}

In the following a \emph{cut vector} is a vector with entries in $\{0, 1\}$, a \emph{cut matrix} is the scaled outer product of two cut vectors (i.e., it is a matrix that is constant on some $S \times T$ block and zero elsewhere), and the \emph{cut norm} $\cutnorm{A}$ of a matrix $A$ is $\max_{v,w} |v^T A w|,$ where the max is taken over cut vectors $v,w$.
We write $\onevec$ for the cut vector where all entries are ones.
In their seminal work on weak algorithmic graph regularity, Frieze and Kannan proved:

\begin{theorem} [Cut Approximation \cite{FK99}] \label{thm:introcut}
For any binary matrix $A \in \rr^{n \times n}$ and $\eps > 0$, there is $\widehat{A}$ that is a sum of $O(\eps^{-2})$ cut matrices such that
$$\cutnorm{A - \widehat{A}} \le \eps n^2.$$
\end{theorem}

The proof of Theorem \ref{thm:introcut} uses a greedy method: in each round we ``select'' a new pair of cut vectors $(v, w)$, and we define $\widehat{A}$ as the projection of $A$ onto the span of the matrices $wv^T$ from pairs chosen so far.
Ideally, one would select cut vectors $v,w$ in each round that achieve the cut norm on $A - \widehat{A}$, but from a computational perspective this task is NP-Hard.
So instead, the selection process typically uses \emph{approximation} algorithms to find $v,w$; when $A$ is dense an extremely fast one was given by Frieze and Kannan \cite{FK99}, and for sparse graphs it is common in the current literature to use a nice approximation algorithm for the cut norm by Alon and Naor \cite{AN04}, based on semidefinite programming.

Later, \Lovasz{} and Szegedy \cite{LS07} proved a highly generalized version of Theorem \ref{thm:introcut}.
They used a slightly different selection rule from \cite{FK99}: in each round they select the cut vectors $(v, w)$ maximizing the mass (measured in Frobenius norm) removed from $A - \widehat{A}$ as the result of their choice.
This method is extremely flexible, and it implies a generalization of Theorem \ref{thm:introcut} to essentially any inner product space (c.f.\ \cite{LS07}, Lemma 4.1).
But the drawback of this method is that it is not algorithmic in the context relevant to graph regularity: it is not clear how to compute (or even approximate) in polynomial time the cut vector pair that removes the most mass from $A - \widehat{A}$.

In this paper, we propose yet another selection rule: in each round, we select the nonzero cut vectors $v,w$ maximizing
\begin{align}
\left| \frac{v^T \left(A - \widehat{A}\right) w}{\|v\|_2 \|w\|_2} \right|. \label{eq:max}
\end{align}
We prove that this approach also leads to an approximation as in Theorem \ref{thm:introcut}.
Additionally, an advantage over the previous two methods is that the maximizing cut vectors $v, w$ can be computed exactly in polynomial time using a simple LP-rounding algorithm by Charikar \cite{Charikar00}.
Our method also generalizes to other inner product spaces in the same way as \cite{LS07}, and in Appendix \ref{app:cutpvdalg} we show that our method remains algorithmic when a diagonal inner product is used.

An important conceptual point is that, since we have a deterministic selection rule, we may treat this process as an \emph{explicit matrix decomposition}.
We define the \emph{$i^{th}$ projection value} $\phi_i$ of the decomposition to be the mass removed from $A$ in the $i^{th}$ round (as measured by Frobenius norm in the ``main'' decomposition, or a different norm when a different inner product is used -- see Section \ref{sec:ddecomp} for details).
There is an analogy here to the singular values, which can be viewed as the mass removed in each round of a similar greedy process where $v, w$ are not constrained to be cut vectors.
The analogy is not perfect, linear algebraically speaking: for example there can be $n^2$ nonzero projection values but there only $n$ singular values, and unlike singular values, the projection values are not necessarily decreasing.
But, as we discuss next, the analogy is useful in the context of regularity lemmas.
Using this, we revisit a recent line of work that proves regularity lemmas by analysis of their singular values, and we show that simpler and stronger proofs can be obtained when projection values are used instead.

\subsection{Regularity for Cut Pseudorandom Graphs}

First, we give a formal background on graph regularity.
In this exposition, let $G$ be an $n$-node undirected unweighted graph with adjacency matrix $A_G$.
The \emph{weak regularity lemma} of Frieze and Kannan states:
\begin{theorem} [Weak Regularity Lemma \cite{FK99}] \label{thm:introweak}
For any $\eps > 0$, there is a vertex partition $\Pi$ of $G$ into $|\Pi| \le 2^{\eps^{-2}}$ parts and a weighted graph $H$ satisfying:
\begin{itemize}
\item The adjacency matrix $A_H$ of $H$ is constant on its $V_1 \times V_2$ block for all $V_1, V_2 \in \Pi$, and
\item $\cutnorm{A_G - A_H} = O(\eps n^2)$.
\end{itemize}
\end{theorem}

It is common to merge the blocks of $\Pi$ into supernodes, and treat $H$ as a node- and edge-weighted graph on a small number of nodes that approximates the cuts of $G$.
In parallel, the \Szemeredi{} regularity lemma controls a more stringent norm, at the price of a worse dependence on $\eps$ in partition size.
For a matrix $A \in \rr^{n \times n}$ and a partition $\Pi$ of $[n]$, let us define the \emph{partition norm} as the quantity
$$\norm{\Pi}{A} := \sum \limits_{V_1, V_2 \in \Pi} \norm{\square}{A[V_1 \times V_2]}$$
where $A[V_1 \times V_2]$ denotes the matrix that agrees with $A$ on its $V_1 \times V_2$ block and is $0$ elsewhere.

\begin{theorem} [\Szemeredi{} Regularity Lemma \cite{Szemeredi75, LS07}] \label{thm:introszem}
For any $\eps > 0$, there is a vertex partition $\Pi$ into $|\Pi| \le \tower(\eps^{-2})$ parts\footnote{Recall that $\tower(x) = 2^{2^{2^{\dots}}}$ with total height $x$.} and a weighted graph $H$ satisfying:
\begin{itemize}
\item The adjacency matrix $A_H$ of $H$ is constant on its $V_1 \times V_2$ block for all $V_1, V_2 \in \Pi$, and
\item $\norm{\Pi}{A_G - A_H} = O(\eps n^2)$.
\end{itemize}
\end{theorem}

We have that $\norm{\square}{A} \le \norm{\Pi}{A}$, and so it is generally stronger to control the partition norm, and this is sometimes required in applications \cite{ACHKRS10, KSSS00}.
For an unweighted adjacency matrix both norms are bounded by $n^2$, so these regularity lemmas improve on the trivial control by a factor of $\eps$.
The norm controls of $\eps n^2$ in Theorems \ref{thm:introweak} and \ref{thm:introszem} apply only when $G = (V, E)$ has unweighted nodes and edges, and they are nontrivial only when $|E| \gg \eps n^2$.
It is natural to ask whether analogous theorems can be proved for sparser graphs.
Unfortunately, there are barriers to achieving analogs that work for \emph{all} sparse input graphs \cite{Gowers97, FL14, MS16}.\footnote{A general sparse regularity lemma is possible if one allows a large ``exceptional set'' in which many edges can be hidden \cite{Scott11, MS19}.  \Szemeredi{}'s original regularity lemma allowed an exceptional set; for dense unweighted input graphs, this is equivalent to the phrasing by partition norm.  For sparser input graphs, the phrasing by partition norm is stronger and often preferable \cite{LS07}, and our exposition in this paper refers to the natural sparse generalization by partition norm.}
However, sparse regularity lemmas can be proved for certain classes of input graphs that satisfy various ``pseudorandomness'' criteria.
We propose a new such graph class here, with the goal of simultaneously generalizing pseudorandomness concepts that have been studied previously in the literature.

In the following, let $G = (V, E, w, \vol)$ be a graph with edge weights ($w$) and node weights ($\vol$).
Let $w^*(u, v) := w(u, v) \vol(u) \vol(v)$, let $w^*(E)$ be the sum of $w^*(e)$ over all edges $e \in E$, and let $A^*_G$ be the \emph{fully-weighted adjacency matrix}: that is, the $(i, j)^{th}$ entry of $A^*_G := w^*(v_i, v_j)$, or $0$ if $(v_i, v_j)$ is a non-edge.
The \emph{cut decomposition} is the greedy matrix decomposition outlined in the previous section, and $\phi^{(r)}$ is the vector that holds the projection values from the first $r$ rounds of this process.
We define:

\begin{definition} [Cut Pseudorandomness] \label{def:pr}
The graph $G$ is \emph{$r$ cut pseudorandom} with respect to a diagonal positive definite matrix $D$ if, when we take the cut decomposition of $A^*_G$ with respect to inner product $\langle x, y \rangle := x^T D y$, the projection values satisfy
$$\left\| \phi^{(r)} \right\|_2 = O\left(\frac{w^*(E)}{\ip{}{\onevec}{\onevec}}\right).$$
\end{definition}

Note that the definition of cut pseudorandomness depends not just on the graph $G$, but also on our selection rule and hence the inner product in play -- a graph that is $r$ cut pseudorandom with respect to one matrix $D$ might not be so under another matrix $D'$.
We then prove the following sparse regularity lemmas for cut pseudorandom graphs:

\begin{theorem} [Regularity Lemmas for Cut Pseudorandom Graphs]
If $G$ is $r = \Theta(\eps^{-2})$ cut pseudorandom, then it has a sparse weak regularity lemma, i.e., Theorem \ref{thm:introweak} holds with $w^*(E)$ in place of $n^2$.
If $G$ is $r = \tower(\eps^{-2})$ cut pseudorandom, then it has a sparse \Szemeredi{} regularity lemma, i.e., Theorem \ref{thm:introszem} holds with $w^*(E)$ in place of $n^2$.
(See Theorems \ref{thm:weakpowreg} and \ref{thm:szempowreg} for self-contained statements.)
\end{theorem}

Moreover, we show that sparse weak and \Szemeredi{} decompositions for cut pseudorandom graphs can be computed in polynomial time.
The point of cut pseudorandomness is not so much that the definition is clearly interpretable, but rather, that it simultaneously captures several previously-studied graph classes that do have clear interpretations.
We discuss these next.

\subsection{Graph Classes Captured by Cut Pseudorandomness}

First, we discuss the following version of \emph{$L^p$ upper regularity}:

\begin{definition} [$L^p$ Upper Regularity \cite{BCCZ14, BCCZ14b}]
The graph $G = (V, E)$ is \emph{$L^p$ upper $\eta$ regular} if, for any vertex partition $\Pi$ of size $|\Pi| \le \eta^{-1}$, we have
$$\left( \sum \limits_{V_1, V_2 \in \Pi} \left(\frac{\vol(V_1)(V_2)}{\vol(V)^2}\right) \left(\frac{w^*(V_1, V_2)}{\vol(V_1)\vol(V_2)}\right)^p\right)^{1/p} \le O\left( \frac{w^*(E)}{\vol(V)^2} \right)$$
where $w^*(V_1, V_2)$ is the sum of $w^*(e)$ over all edges going between $V_1$ and $V_2$.
\end{definition}

The original definition of $L^p$ upper regular considers a somewhat narrower class of partitions; namely, those that satisfy $\vol(V_i) \ge \eta \vol(V)$ for all $V_i \in \Pi$.
We discuss the significance of this change in Section \ref{sec:lp}.
$L^p$ upper regular graphs are important mainly because, in \cite{BCCZ14, BCCZ14b}, the authors extend a theory of graph limits from the case $L^{\infty}$ to general $p$; their interpretation is further discussed in these papers.
The case $p = \infty$ was first explored by Kohayakawa \cite{Kohayakawa97} and R{\" o}dl \cite{Rodl73}, and the version with node weights was considered in \cite{ACHKRS10}.
These papers develop sparse \Szemeredi{} regularity lemmas for $L^{\infty}$ upper regular graphs, although their lemmas do not readily extend to $L^p$ upper regular graphs when $p < \infty$.
$L^{\infty}$ upper regular graphs include the notable special case of \emph{jumbled graphs} \cite{Thomason87, Thomason87b, CGW89, KS06}, which in turn include \emph{$(n, d, \lambda)$ graphs} with appropriate parameters; these cases are studied in particular because they enjoy some expanded applications over the general case \cite{GS05, CFZ12}.

%
%
%
%

Alongside this theory, there are also sparse regularity lemmas for different spectrally-motivated classes of input graphs.
These include the following.
\begin{definition} [Low Threshold Rank \cite{BRS11, GS11, GS12, GS13, GharanT15}]
For a graph $G$, let $A$ be its adjacency matrix, $D$ its diagonal matrix of node degrees, and $\overline{A} = D^{-1/2} A D^{-1/2}$ its normalized adjacency matrix.
For $\eps > 0$, the \emph{$\eps$ threshold rank} of $G$ is
$$t_{\eps}(G) := \sum \limits_{\lambda \text{ an eigenvalue of } \overline{A}, \lambda > \eps} \lambda^2.$$
We say that $G$ has \emph{low $\eps$-threshold rank} if $t_{\eps}(G) = O(1)$.
\end{definition}

We also discuss the ``weighted'' version of this definition in Section \ref{sec:ltr}.
Graphs of low threshold rank gained prominence due to a sequence of papers achieving approximation algorithms for these graphs, where comparable approximation algorithms are known to be hard in the general case (e.g., \cite{BRS11, GS11, GS12, GS13}).
Oveis Gharan and Trevisan \cite{GharanT15} gave a nice explanation for these by supplying a weak regularity lemma for low threshold rank graphs, and using it for further improved approximation algorithms.
A related property was considered by Kannan and Vempala \cite{KV09}:

\begin{definition} [Core Density \cite{KV09}]
The core strength of a graph $G$ with adjacency matrix $A$ is the quantity
$$\sum \limits_{i, j} \frac{A_{ij}}{\left(\deg(i) + \overline{d}\right)\left(\deg(j) + \overline{d}\right)}$$
where $\overline{d}$ is the average node degree.
We say that $G$ is \emph{core dense} if its core strength is $O(1)$.
\end{definition}
This also has a natural weighted extension, discussed in Section \ref{sec:ltr}.
The main advantage of core density is that it can additionally be extended to tensors; we refer to \cite{KV09} for details.
They developed a weak regularity lemma for these graphs, and similarly used it towards approximation algorithms for core dense graphs.

Currently, all of the above graph classes have related but somewhat different proofs of weak regularity lemmas.
Our point is that all of these graph classes can be viewed as special cases of cut pseudorandomness:

\begin{theorem} [Cut Pseudorandomness of Graph Classes] ~
\begin{itemize}
\item If a graph is $L^p$ upper $\eta$-regular for any $p \ge 2$ (with the change in definition mentioned above), then it is $O(\log(1/\eta))$-cut-pseudorandom.

\item If a graph has low $\eps$ threshold rank, then it is $O(1/\eps^2)$-cut-pseudorandom.

\item If a graph is core dense, then it is $r$-cut-pseudorandom for every $r$.
\end{itemize}
\end{theorem}

Hence we get sparse weak \emph{and \Szemeredi{}} regularity lemmas for these graphs in one shot; except in the case of $L^{\infty}$ upper regularity, the sparse \Szemeredi{} regularity lemmas are new.
We remark that \cite{BCCZ14, BCCZ14b} establish weak regularity lemmas for $L^p$ upper regular graphs for all $p \ge 1$, so we leave a small gap here, since our theorem begins at $p\ge 2$.
References are given in Table \ref{tbl:sparsereg}.

\renewcommand{\arraystretch}{1.2}
\begin{table}
\begin{center}
\begin{tabular}{lcc}
\toprule
Sparse Graph Class & Sparse \Szemeredi{} Regularity? & Sparse Weak Regularity? \\
\midrule
$L_{\infty}$ Upper Regular & \checkmark{} \cite{Kohayakawa97, Rodl73} & \checkmark{} \cite{Kohayakawa97, Rodl73}  \\
$L_{p}$ Upper Regular & \checkmark{}\checkmark{} & \checkmark{} \cite{BCCZ14, BCCZ14b}  \\
Low Threshold Rank & \checkmark{}\checkmark{} & \checkmark{} \cite{GharanT15} \\
Core Dense & \checkmark{}\checkmark{} & \checkmark{} \cite{KV09} \\
\bottomrule
\end{tabular}
\end{center}
\caption{\label{tbl:sparsereg} Sparse regularity lemmas for various graph classes in prior work are indicated with single checkmarks.  We give a unified proof of all entries in this table, thus recovering all checkmarks and proving the double checkmarks as new results.}
\end{table}

\subsection{Algorithmic Applications}

Regularity lemmas are commonly used in approximation algorithms of NP-hard problems that can be related in some way to graph cuts.
For example, it is NP-Hard to find a MAX-CUT PTAS for general graphs \cite{KKMO04}, but the weak regularity lemma of Frieze and Kannan \cite{FK99} was initially used for an extremely efficient MAX-CUT PTAS for dense unweighted input graphs.
A natural direction is then to determine what other natural classes of graphs admit a PTAS for MAX-CUT or similar problems.
Oveis Gharan and Trevisan \cite{GharanT15} proved PTASes for MAX-CUT, MAX-BISECTION, and MIN-BISECTION for graphs of low threshold rank, following prior work in \cite{BRS11, GS11, GS12, GS13}.
In fact, they essentially prove more generally that any graph class admitting suitable algorithmic regularity lemmas have PTASes for these problems.
Using their work, we get the following algorithmic corollaries of our decomposition:

\begin{corollary} \label{cor:matrixalgs}
For any $\eps > 0$ and any $n$-node unweighted graph $G = (V, E)$ that is $O(\eps^{-2})$ cut pseudorandom (with respect to any matrix $D$), one can solve any of the following problems in $2^{\widetilde{O}(1/\eps^3)} + \poly(n)$ time:
\begin{itemize}
\item MAX-CUT within $\pm \eps |E|$ error,
\item MAX-BISECTION within $\pm \eps |E|$ error, and
\item MIN-BISECTION within $\pm \eps |E|$ error.
\end{itemize}
\end{corollary}

For example, this implies that $L_p$ upper regular graphs (and others) have PTASes for all these problems, whereas none was previously known.
These results are all discussed in more detail in Section \ref{sec:algorithms}.

%
%
%

\section{Cut Decompositions}

The \emph{Frobenius inner product} of two matrices $\ip{F}{\cdot}{\cdot}$ is defined by
$$\ip{F}{A}{B} := \tr(AB^T),$$
which may be equivalently viewed as the Euclidean inner product of $A, B$ as vectors.
The \emph{Frobenius norm} $\norm{F}{A}$ is the corresponding matrix norm, which may also be viewed either as the $L_2$-norm of the entries or of the singular values of $A$.
The definitions of cut vectors, cut matrices, cut norm, partition norm, and the $\onevec$ vector will continue to be used in the following technical content; see Section \ref{sec:cutdecomp} for a reminder of these definitions.

\subsection{Basic Cut Decompositions}

As a warmup, we will first present the ``basic'' cut decomposition.
The next part presents a generalization that is needed for a few applications later in the paper.

Let $A \in \rr^{n \times n}$, and let $S = \emptyset$ be an initially empty subset of cut matrices (recall that a cut matrix is a scaled outer product of two cut vectors).
At all times, the matrix $\widehat{A}$ is defined as the projection, by $\ip{F}{\cdot}{\cdot}$, of $A$ onto $\spann(S)$.
Since $S$ is initially empty, $\widehat{A}$ is initially $0$.

\paragraph{Decomposition Algorithm.}
Repeat the following process until $\widehat{A} = A$.
Let
$$v^*, w^* := \arg \max \limits_{v, w \ne 0 \in X} \left| \frac{v^T}{\|v\|_2} \left(A - \widehat{A}\right) \frac{w}{\|w\|_2} \right|,$$
breaking ties arbitrarily, and add the cut matrix $w^* v^{*T}$ to $S$.
Note that $v^*, w^*$ can be computed by a simple polynomial time algorithm of Charikar \cite{Charikar00}; see Appendix \ref{app:cutpvdalg} for more details.
We repeat the process at most $n^2$ times, since $\rr^{n \times n}$ is an $n^2$ dimensional space.

\paragraph{Definitions and Properties.}
Let $\widehat{A}_i$ be the value of $\widehat{A}$ after the $i^{th}$ round of this process, and for $i \ge 1$ let $A_i := \widehat{A}_{i} - \widehat{A}_{i-1}$.
The \emph{$i^{th}$ cut projection value} is the quantity $\phi_i := \frobnorm{A_i}$.
Throughout, we use the notation $\phi^{(j)}$ to denote the vector containing the first $j$ projection values.
Finally, we note that by construction the matrices $\{A_i\}$ are pairwise orthogonal (under $\ip{F}{\cdot}{\cdot}$), and that each $\widehat{A}_i$ is a linear combination of at most $i$ cut matrices.

\subsection{Cut Decompositions under Diagonal Inner Products \label{sec:ddecomp}}

Let $D \in \rr^{n \times n}$ be a diagonal positive definite matrix.
The cut decomposition of a matrix $A$ \emph{with respect to $D$} is defined as follows.
The basic cut decomposition presented above is the same as the cut decomposition of $A$ with respect to the identity matrix $I$.
First, the inner product $\ip{D}{\cdot}{\cdot}$ is defined by
$$\ip{D}{A}{B} := \ip{F}{D^{1/2} AD^{1/2}}{D^{1/2}BD^{1/2}},$$
and $\norm{D}{\cdot}$ is the corresponding matrix norm.
Again, $S$ is an initially-empty set of cut matrices.
At all times, $\widehat{A}$ is the projection of the matrix $D^{-1}AD^{-1}$ onto the span of the matrices in $S$, and now the projection is by the inner product $\ip{D}{\cdot}{\cdot}$.
In each round, we select the cut matrix $M$ satisfying
\begin{align}
M := \arg \max \limits_{M^* \text{ cut matrix}} \left| \ip{D}{\frac{M^*}{\norm{D}{M^*}}}{D^{-1}AD^{-1} - \widehat{A}} \right|, \label{eq:vwopt}
\end{align}
and we add the cut matrix $M$ to $S$.
Roughly as before, we define:
\begin{itemize}
\item $\widehat{A}_i$ is the value of $\widehat{A}$ after $i$ rounds (so $\widehat{A}_0 = 0$),
\item $A_i := \widehat{A}_i - \widehat{A}_{i-1}$, and
\item the $i^{th}$ projection value $\phi_i$ is $\norm{D}{A_i}$.
\end{itemize}

The first thing to point out is that the algorithm of Charikar \cite{Charikar00} no longer directly finds the cut matrix $M$ maximizing (\ref{eq:vwopt}) in polynomial time, essentially due to the renormalization of $M^*$ by $\norm{D}{\cdot}$.
Thus, in Appendix \ref{app:cutpvdalg} we show that Charikar's algorithm can be straightforwardly extended to solve (\ref{eq:vwopt}), so long as the matrix $D$ holds integer entries on its diagonal between $1$ and $\poly(n)$ (which it does for all cases of interest in this paper).

\subsection{Projection Values vs.\ Matrix Norms}

Here we prove some inequalities that relate the projection values of the decomposition to various matrix norms of interest.
In the following lemmas let $A, D$ be matrices where $D$ is diagonal positive definite, and suppose we decompose $A$ with respect to $D$ giving projection values $\{\phi_i\}$.
\begin{lemma} \label{lem:cutbound}
For all indices $0 \le i < n^2$, we have
$\left(\onevec^T D \onevec\right) \phi_{i+1} \ge \norm{\square}{A - D\widehat{A}_i D}.$
\end{lemma}
\begin{proof}
We have
\begin{align*}
\phi_{i+1} = \norm{D}{A_{i+1}} &= \norm{F}{D^{1/2}\left(\widehat{A}_{i+1} - \widehat{A}_i\right)D^{1/2}}\\
&\ge \norm{2}{D^{1/2}\left(\widehat{A}_{i+1} - \widehat{A}_i\right)D^{1/2}}\\
&\ge \max \limits_{v, w \in X} \left|\frac{v^T D^{1/2}}{\norm{2}{D^{1/2} v}} D^{1/2} \left( \widehat{A}_{i+1} - \widehat{A}_i \right) D^{1/2} \frac{D^{1/2} w}{\norm{2}{D^{1/2} w}}\right|\\
&= \max \limits_{v, w \in X} \left| \ip{D}{\frac{wv^T}{\norm{D}{wv^T}}}{\widehat{A}_{i+1} - \widehat{A}_i}\right|.
\end{align*}
Recall that we have defined $\widehat{A}_{i+1}$ as the projection of $D^{-1}AD^{-1}$, by $\ip{D}{\cdot}{\cdot}$, onto the span of a set of cut matrices, including $wv^T$ where $v,w$ maximize the above expression.
Thus we may continue
\begin{align*}
&= \max \limits_{v, w \in X} \left| \ip{D}{\frac{wv^T}{\norm{D}{wv^T}}}{D^{-1} A D^{-1} - \widehat{A}_i}\right|\\
&= \max \limits_{v, w \in X} \left|\frac{v^T D^{1/2}}{\norm{2}{D^{1/2} v}} D^{1/2} \left( D^{-1} A D^{-1} - \widehat{A}_i \right) D^{1/2} \frac{D^{1/2} w}{\norm{2}{D^{1/2} w}}\right|\\
&= \max \limits_{v, w \in X} \left|\frac{v^T \left(A - D \widehat{A}_i D\right) w}{\sqrt{v^T D v} \sqrt{w^T D w}} \right|\\
&\ge \frac{\cutnorm{A - D\widehat{A}_i D}}{\onevec^T D \onevec}. \tag*{\qedhere}
\end{align*}
\end{proof}

\begin{lemma} \label{lem:partbound}
For any indices $i \le j$ and any partition $\Pi$ of $[n]$, we have
$$\norm{\Pi}{ D\left(\widehat{A}_j - \widehat{A}_i\right) D} \le \left(\onevec^T D \onevec\right) \norm{2}{\phi(i::j)}.$$
\end{lemma}
\begin{proof}
First, we notice that $\norm{\Pi}{ D\left( \widehat{A}_j - \widehat{A}_i\right)D}$ is maximized when $\Pi$ is the completely refined partition that puts each index into its own singleton subset.
So without loss of generality we may prove the inequality in this setting.
In the following, absolute value symbols on matrices $|M|$ are applied entrywise, and we will use the fact that they commute with diagonal matrices; for example, $D|M| = |DM|$.
\begin{align*}
\norm{\Pi}{ D\left(\widehat{A}_j - \widehat{A}_i\right) D} &= \left(\sum \limits_{a,b=1}^n \left| \left(D\left(\widehat{A}_j - \widehat{A}_i\right) D\right)_{a,b} \right|\right)\\
&= \ip{F}{\onevec \onevec^T}{\left|D\left(\widehat{A}_j - \widehat{A}_i\right) D\right|}\\
&= \ip{F}{D^{1/2} \onevec \onevec^T D^{1/2}}{ \left|D^{1/2} \left(\widehat{A}_j - \widehat{A}_i\right) D^{1/2}\right| } \tag*{$D$ diagonal}\\
&\le \norm{F}{D^{1/2} \onevec \onevec^T D^{1/2}}\norm{F}{D^{1/2}
\left( \widehat{A}_j - \widehat{A}_i\right) D^{1/2}} \tag*{Cauchy-Schwarz}\\
&=  \left( \onevec^T D \onevec \right) \norm{D}{\widehat{A}_j - \widehat{A}_i}\\
&=  \left( \onevec^T D \onevec \right) \norm{D}{\sum \limits_{k=i+1}^j A_k}.
\end{align*}
Next, since the matrices $\{A_k\}$ are pairwise orthogonal under $\ip{D}{\cdot}{\cdot}$, we have
\begin{align*}
\norm{D}{\sum \limits_{k=i+1}^j A_k}^2 = \sum \limits_{k=i+1}^j \norm{D}{A_k}^2 = \norm{2}{\phi(i::j)}^2
\end{align*}
which combined with the above, gives
\begin{align*}
\norm{\Pi}{D\left(\widehat{A}_j - \widehat{A}_i\right)D} \le \left( \onevec^T D \onevec \right) \norm{2}{\phi(i::j)}. \tag*{\qedhere}
\end{align*}
\end{proof}

\begin{lemma} \label{lem:dinvbound}
For any index $r$, we have
$$\norm{2}{\phi^{(r)}} \le \norm{2}{\Lambda^{(r)}},$$
where $\Lambda$ is the vector containing the singular values of the matrix $D^{-1/2} A D^{-1/2}$ in non-increasing order.
\end{lemma}
\begin{proof}
First, we have
\begin{align*}
\norm{2}{\phi^{(r)}}^2 &= \sum \limits_{i=1}^r \norm{D}{A_i}^2\\
&= \norm{D}{\sum \limits_{i=1}^r A_i}^2 \tag*{$\{A_i\}$ orthogonal}\\
&= \norm{D}{\widehat{A}_r}^2\\
&= \norm{F}{D^{1/2} \widehat{A}_r D^{1/2}}^2.
\end{align*}
Recall that $\widehat{A}_r$ is the projection of $D^{-1}AD^{-1}$, by $\ip{D}{\cdot}{\cdot}$, onto the span of $r$ cut matrices, which all have rank one.
Equivalently, we can say that $D^{1/2} \widehat{A}_r D^{1/2}$ is the projection of the matrix $D^{-1/2} A D^{-1/2}$, by standard Frobenius inner product $\ip{F}{\cdot}{\cdot}$, onto the span of $r$ matrices of rank one.
Let $M^*$ be the matrix maximizing $\norm{F}{M^*}$ over the matrices that are projections of $D^{-1/2} A D^{-1/2}$ onto the span of $r$ rank one matrices; thus, we have
$$\norm{F}{M^*} \ge \norm{F}{D^{-1/2} \widehat{A}_r D^{-1/2}}.$$
Since $M^*$ is found by taking the leading $r$ terms of the singular value decomposition of $D^{-1/2} A D^{-1/2}$, it follows that $\norm{F}{M^*}$ is exactly the $2$-norm of the leading $r$ singular values.
Thus
\begin{align*}
\norm{2}{\phi^{(r)}}^2 = \norm{F}{D^{1/2} \widehat{A}_r D^{1/2}}^2 \le \norm{F}{M^*}^2 = \norm{2}{\Lambda^{(r)}}^2. \tag*{\qedhere}
\end{align*}
\end{proof}

\section{Sparse Regularity Lemmas for Cut Pseudorandom Graphs}

We will prove various regularity lemmas for \emph{cut pseudorandom graphs} in this section.
Let us first introduce some notation.
The graphs $G = (V, E)$ discussed in the rest of the paper are undirected and have both edge weights and node weights, all of which are positive.
In principle we can allow self-loops if desired.
\begin{itemize}
\item We write $\vol(v)$ for the weight of a node $v$, or $\vol(V)$ for the summed weight of the node set $V$.

\item We write $w(e)$ for the weight of an edge $e$, or $w(E)$ for the sum of $w(e)$ over all $e \in E$, or $w(V_1, V_2)$ for the sum of $w(e)$ over all $e \in E \cap (V_1 \times V_2)$.

\item We write $w^*(u, v) := w(u, v) \vol(u) \vol(v)$, and similarly $w^*(E)$ is the sum of $w^*(e)$ over all edges $e \in E$, and $w^*(V_1, V_2)$ is the sum of $w^*(e)$ over all edges $e \in E \cap (V_1 \times V_2)$.

\item The \emph{edge-weighted adjacency matrix $A$} has entries $A_{ij} = w(v_i, v_j)$ for each edge $(v_i, v_j)$, or $A_{ij} = 0$ if $(v_i, v_j)$ is a non-edge.
The \emph{fully-weighted adjacency matrix $A^*$} instead has entries $A^*_{ij} = w^*(v_i, v_j)$, or $A^*_{ij} = 0$ if $(v_i, v_j)$ is a non-edge.
\end{itemize}
 
\begin{definition} [Cut Pseudorandomness]
We say that a graph $G$ is \emph{$r$ cut pseudorandom} with respect to a diagonal positive definite matrix $D$ if, when we take the cut decomposition of its fully-weighted adjacency matrix $A^*$ with respect to $D$, the projection values $\{\phi_i\}$ satisfy
$$\norm{2}{\phi^{(r)}} = O\left( \frac{w^*(E)}{\onevec^T D \onevec} \right).$$
We will sometimes just say ``cut pseudorandom,'' making $D$ implicit, when we do not need to reference it.
\end{definition}

We notice that cut pseudorandomness implies $r' < r$ cut pseudorandomness, so it is typically strongest to state our results with $r$ as small as possible.


\subsection{Sparse Cut Approximation and Sparse Weak Regularity}

Here we prove two closely-related results on approximating the cut norm of a graph:

\begin{lemma} [Sparse Cut Approximation] \label{lem:sparsecut}
Let $\eps > 0$ and let $G = (V, E)$ be $r = \Theta(\eps^{-2})$ cut pseudorandom.
Then there is a graph $H = (V, E_H)$, possibly with different node and edge weights from $G$, such that:
\begin{itemize}
\item The edge-weighted adjacency matrix of $H$ is the sum of $O(\eps^{-2})$ cut matrices, and
\item $\norm{\square}{A^*_G- A^*_H} = O(\eps w^*(E))$, where $A^*_G, A^*_H$ are the fully-weighted adjacency matrices of $G, H$ respectively.
\end{itemize}
\end{lemma}
\begin{proof}
Compute a cut decomposition of the fully-weighted adjacency matrix $A^*_G$ of $G$ with respect to the matrix $D$ implying cut pseudorandomness.
Let $\{\phi_i\}$ be the projection values.
The smallest entry $\phi_i$ in $\phi^{(r)}$ satisfies
$$\phi_i \le \frac{\norm{2}{\phi^{(r)}}}{\sqrt{r}} = O\left(\eps \norm{2}{\phi^{(r)}}\right) = O\left( \eps \cdot \frac{w^*(E)}{\onevec^T D \onevec}\right),$$
where the last equality follows by cut pseudorandomness of $G$.
By Lemma \ref{lem:cutbound}, we have
$$\phi_i \ge \frac{\norm{\square}{A^*_G - D \widehat{A}^*_{i-1} D}}{\onevec^T D \onevec},$$
where $\widehat{A}^*_{i-1}$ is from the decomposition of $A^*_G$.
Combining these inequalities, we have
\begin{align*}
\norm{\square}{A^*_G - D\widehat{A}^*_{i-1} D} = O(\eps w^*(E)).
\end{align*}
Then, we recall by construction that $\widehat{A}^*_{i-1}$ is a linear combination of $i-1 \le r$ cut matrices.
We interpret this as the edge-weighted adjacency matrix of the graph $H$, and we interpret the diagonal entries of the matrix $D$ as the node weights of $H$.
Thus $D \widehat{A}^*_{i-1} D = A^*_H$, and the lemma follows.
\end{proof}

\begin{theorem} [Sparse Weak Regularity Lemma] \label{thm:weakpowreg}
Let $\eps > 0$ and let $G = (V, E)$ be $r = \Theta(\eps^{-2})$ cut pseudorandom.
Then there is a graph $H = (V, E_H)$, possibly with different node and edge weights from $G$, such that:
\begin{itemize}
\item There is a vertex partition $\Pi$ into $|\Pi| = 2^{O(\eps^{-2})}$ parts such that for all $V_i, V_j \in \Pi$ the edge-weighted adjacency matrix of $H$ is constant on its $V_i \times V_j$ block, and
\item $\norm{\square}{A^*_G- A^*_H} = O(\eps w^*(E))$, where $A^*_G, A^*_H$ are the fully-weighted adjacency matrices of $G, H$ respectively.
\end{itemize}
\end{theorem}
\begin{proof}
This follows almost immediately from Lemma \ref{lem:sparsecut}.
We simply recall that the edge-weighted adjacency matrix of $H$ in Lemma \ref{lem:sparsecut} is the sum of $\le r$ cut matrices, and we notice that each cut matrix is the outer product of two cut vectors, which each define a bipartition of the vertex set.
Thus, the common refinement $\Pi$ of all these bipartitions has $|\Pi| \le 2^{2r} = 2^{O(r)}$ parts, and the edge-weighted adjacency matrix of $H$ is constant on its $V_i \times V_j$ block for all $V_i, V_j \in \Pi$.
\end{proof}

We remark here that it is common to merge together the parts in $\Pi$ into supernodes, setting the weight of each supernode to be the sum of the node weights that comprise it.
In this sense one can view $H$ as a graph on very few nodes that approximates $G$ in cut norm.
We will stay away from this phrasing, however, as it requires one to formalize the embedding of the nodes of $H$ into the nodes of $G$ which can get cumbersome.

\subsection{Sparse \Szemeredi{} Regularity}

%
%
%
%

We now prove our sparse \Szemeredi{} regularity lemma.
Our approach is a thematic expansion on ideas from \cite{Szegedy10, LS07}, which give proofs of regularity lemmas based on the SVD.
We show that our projection values can be effectively substituted for singular values in a way that makes the approach work for sparse regularity as well.
This particular exposition is influenced by \cite{Taoblog}.
\begin{theorem} [Sparse \Szemeredi{} Regularity Lemma] \label{thm:szempowreg}
Let $\eps > 0$ and let $G = (V, E)$ be $r = \tower(\Theta(\eps^{-2}))$ cut pseudorandom.
Then there is a graph $H = (V, E_H)$, possibly with different node and edge weights from $G$, and a vertex partition $\Pi$ into $|\Pi| = \tower(O(\eps^{-2}))$ parts such that:
\begin{itemize}
\item For all $V_i, V_j \in \Pi$ the edge-weighted adjacency matrix of $H$ is constant on its $V_i \times V_j$ block, and
\item $\norm{\Pi}{A^*_G- A^*_H} = O(\eps w^*(E))$, where $A^*_G, A^*_H$ are the fully-weighted adjacency matrices of $G, H$ respectively.
\end{itemize}
\end{theorem}
\begin{proof}

Let $G$ be $r$ cut pseudorandom with respect to the matrix $D$.
We will treat the parameter $r$ as an integer that can be selected later in the proof, verifying at the end that we only need to use a bound of the form $r = \tower(\Omega(\eps^{-2}))$.
Let $q$ be an integer and let $f : \mathbb{N} \to \mathbb{N}$ be a function satisfying $f(n) > n$ for all $n$, both of which we will also choose later in the proof.
Compute the cut decomposition of $A^*_G$ with respect to $D$.
By Lemma \ref{lem:sparsecut}, so long as $r \ge f(q)$, there exists an index $j \le f(q)$ such that
\begin{align}
\norm{\square}{A^*_G - D\widehat{A}^*_j D} = O\left(\frac{w^*(E)}{f(q)^{1/2}}\right), \label{eq:cutnormbound}
\end{align}
where $\widehat{A}^*_j$ is from the cut decomposition.
Additionally, let $i := \min\{q, j\}$.
As before, $\widehat{A}^*_i$ (also from the cut decomposition) is a linear combination of $i$ cut matrices, each of which is formed using two cut vectors, which each define two bipartitions of the vertices.
Let $\Pi$ be the common refinement of these bipartitions, which thus has $|\Pi| \le 2^{O(i)}$ parts, and we note that $\widehat{A}^*_i$ is constant on each $V_1 \times V_2$ block for $V_1, V_2 \in \Pi$.
We define $H$ to be the graph whose node weights are given by $D$ and whose edge-weighted adjacency matrix is given by $\widehat{A}^*_i$; thus $H$ satisfies the first point in the theorem.
Our goal is now to bound $\norm{\Pi}{A^*_G - A^*_H}$.
Using the triangle inequality we can write
$$\norm{\Pi}{A^*_G - A^*_H} = \norm{\Pi}{A^*_G - D\widehat{A}^*_i D} \le \norm{\Pi}{A^*_G - D\widehat{A}^*_j D} + \norm{\Pi}{D\widehat{A}^*_j D - D \widehat{A}^*_i D}.$$
So it suffices to bound each of the two norms on the right-hand side as $O(\eps w^*(E))$.

%

\paragraph{The First Norm.}

To bound $\norm{\Pi}{A^*_G - D \widehat{A}^*_j D}$, we leverage our freedom to choose $f$ and (partially) our freedom to choose $q$.
We have:
\begin{align*}
\norm{\Pi}{A^*_G - D \widehat{A}^*_j D} &= \sum \limits_{V_1, V_2 \in \Pi} \max \limits_{S \subseteq V_1, T \subseteq V_2} \left| \left(A^*_G - D\widehat{A}^*_j D\right)(S, T) \right|\\
&\le 2^{O(q)} \norm{\square}{A^* - D\widehat{A}^*_j D}\\
&\le \frac{2^{O(q)}}{f(q)^{1/2}} w^*(E)
\end{align*}
where the second-to-last inequality is by unioning over the $\le 2^{O(q)}$ pairs of parts in the partition, and the last inequality is from (\ref{eq:cutnormbound}). 
We now choose $f$ to be an exponential function of the form $f(q) = C^{q}$ with a large enough base $C$, and we enforce a lower bound of $q \ge \log (1/\eps)$, giving
\begin{align*}
\norm{\Pi}{A^*_G - D \widehat{A}^*_j D} &\le \frac{2^{O(q)}}{f(q)^{1/2}} w^*(E)\\
&\le w^*(E) \cdot 2^{-O(q)} \\
&\le O(\eps w^*(E)) \tag*{(first norm control).}
\end{align*}

\paragraph{The Second Norm.} We next bound $\norm{\Pi}{D\widehat{A}^*_j D - D \widehat{A}^*_i D}$.
In the case where $j \le q$, we have $i = j$ and so $\widehat{A}^*_i = \widehat{A}^*_j$ and so this term is identically zero.
So we will assume in the following that $i = q < j \le f(q)$.
The idea here is to consider the projection values $\{\phi_i\}$ from the decomposition, and choose the value of $q$ to be sure that not too much of the total mass of the projection values lies between $\phi_{i+1}$ and $\phi_j$, and then we can apply Lemma \ref{lem:partbound} to convert this to a bound on the partition norm.
Let $f^{(i)}$ denote the function $f$ iterated $i$ times, and suppose
$r \ge f^{\Omega(\eps^{-2})}(0)$.
Letting $\{\phi_i\}$ be the projection values of the decomposition, and denote by $\phi(x :: y)$ the vector holding the subsequence of projection values $(\phi_x, \dots, \phi_y)$.
By the pigeonhole principle there is an integer $\log(1/\eps) \le k = O(\eps^{-2})$ satisfying\footnote{The lower bound $k \ge \log(1/\eps)$ is only used to enforce the lower bound $q \ge \log(1/\eps)$ picked up in the analysis of the first norm; it does not play a role in bounding the second norm.}
\begin{align}
\norm{2}{\phi\left(f^{(k)}(0) +1:: f^{(k+1)}(0)\right)}^2 = O\left( \eps^2 \norm{2}{\phi^{(r)}}^2\right). \label{eq:midfrob}
\end{align}
We set $q := f^{(k)}(0)$ according to this index $k$ satisfying (\ref{eq:midfrob}).
Now we bound:
\begin{align*}
\norm{\Pi}{D \widehat{A}^*_j D - D \widehat{A}^*_i D} &\le \left(\onevec^T D \onevec\right) \norm{2}{\phi(i::j)} \tag*{Lemma \ref{lem:partbound}}\\
&\le \left(\onevec^T D \onevec\right) \cdot \norm{2}{\phi\left( i :: f^{(i)} \right)} \tag*{$j \le f(i)$}\\
&= \left(\onevec^T D \onevec\right) \cdot O\left(\eps \norm{2}{\phi^{(r)}}\right) \tag*{using (\ref{eq:midfrob})}\\
&= O(\eps w^*(E)) \tag*{(second norm control)}
\end{align*}
where the last step follows from the assumption that $G$ is cut pseudorandom.

\paragraph{Parameter Checks.}

Our control on the first term requires us to choose $f$ to be an exponential function.
In the second norm control, we set $q = f^{O(\eps^{-2})}(0)$, and thus $q = \tower(O(\eps^{-2}))$.
The partition size is $|\Pi| = 2^{O(q)} = \tower(O(\eps^{-2}))$ as well.
We pick up a lower bound of $r \ge f(q)$ in the setup, and $r \ge f^{\Omega(\eps^{-2})}(0)$ in the second norm control; thus both are satisfied so long as $r = \tower(\Omega(\eps^{-2}))$.
\end{proof}

\section{Cut Pseudorandomness of Some Graph Classes}

In the following we will consider some graph classes studied in prior work on sparse regularity and show that they are cut pseudorandom.

\subsection{$L^p$ Upper Regularity \label{sec:lp}}


In this part, we will discuss a slightly modified version of \emph{$L^p$ upper regularity}:
\begin{definition} [$L^p$ Upper Regularity \cite{BCCZ14, BCCZ14b}] \label{def:upperreg}
A graph $G = (V, E)$ is \emph{$L^p$ upper regular} with parameters $(C, \eta)$ if, for any vertex partition $\Pi$ of size $|\Pi| \le \eta^{-1}$, we have
$$\left( \sum \limits_{V_1, V_2 \in \Pi} \frac{\vol(V_1) \vol(V_2)}{\vol(V)^2} \cdot \left(\frac{w^*(V_1, V_2)}{\vol(V_1)\vol(V_2)}\right)^p \right)^{1/p} \le C \frac{w^*(E)}{\vol(V)^2} .$$
We also allow $p = \infty$, in which case we amend the above sum to a max in the usual way, giving the condition
$$\max \limits_{\substack{V_i, V_j \subseteq V,\\ \vol(V_i), \vol(V_j) \ge \eta \vol(V)}} \frac{w^*(V_1, V_2)}{\vol(V_1)\vol(V_2)} \le C\frac{w^*(E)}{\vol(V)^2} .$$
\end{definition}

The original definition considers partitions where $\vol(V_i) \ge \eta \vol(V)$ for all parts $V_i \in \Pi$.
Our change to possibly unbalanced partitions is in line with the take on regularity throughout this paper inspired by \cite{FL14, LS07}, which de-emphasizes the need for balance between different parts in the partition.
It is very likely possible to enforce balance by adding some massaging steps to the decomposition itself; for examples of this in prior work, see \cite{ACHKRS10} (c.f.\ Algorithm 4.3, Step 6.2, and its proof of correctness), or \cite{BCCZ14} (c.f.\ Lemmas 3.1, 3.2, 3.3).
However, we have not included similar massaging here.

This change is definition is rather restrictive for $p=\infty$: for example, supposing $G$ has unit node and edge weights, the left-hand side is maximized at $1$ by choosing $V_1, V_2$ to be singleton sets containing nodes connected by an edge, and thus it forces $|E| \ge Cn^2$, hence requiring the graph to be dense in the classical sense.
At the other extreme, the change does not matter at all when $p=1$: all graphs are $L^1$ upper regular in either definition.
As we will see shortly, the changed definition only needs to be applied at the case $p=2$, which we believe is more towards the unrestrictive end.

\begin{theorem} \label{thm:lpdense}
For any $\eta > 0$, $p \ge 2$, every $G$ that is $L^p$ upper regular with parameters $(\eta, O(1))$ is also $r = \Theta(\log 1/\eta)$ cut pseudorandom.
\end{theorem}
\begin{proof}
First, as shown in \cite{BCCZ14, BCCZ14b}, H{\"o}lder's inequality implies that $L^p$ upper regularity implies $L^{p' < p}$ upper regularity, so it suffices to prove the claim for $p=2$ only.
To do so, let $A^*, D$ be the fully-weighted adjacency matrix and node weight matrix of $G$, respectively.
Take the cut decomposition of $A^*$ with respect to $D$, and let $\Pi$ be the common refinement of all the $\le 2r$ bipartitions used in the leading $r$ cut matrices of the decomposition, which have the matrix $\widehat{A}^*_r$ from the decomposition in their span.
Thus we have $|\Pi| < 1/\eta$, assuming $r = O(\log 1/\eta)$ with a small enough implicit constant.

By construction, the matrix $\widehat{A}^*_r$ is the projection of the matrix $D^{-1}A^* D^{-1}$, by $\ip{D}{\cdot}{\cdot}$, onto the span of the first $r$ selected cut matrices.
Let $\widehat{B}^*_r$ be the projection of $D^{-1}A^* D^{-1}$, by $\ip{D}{\cdot}{\cdot}$, onto the span of all cut matrices corresponding to parts in $\Pi$ (that is, all cut matrices of the form $v_1 v_2^T$, where $v_1, v_2$ are the indicator vectors of some $V_1, V_2 \in \Pi$).
Notice that the span of these cut matrices from $\Pi$ contains the span of the leading $r$ cut matrices in the decomposition, and thus we have
$$\norm{D}{\widehat{A}^*_r} \le \norm{D}{\widehat{B}^*_r}.$$
For all $V_1, V_2 \in \Pi$, we have that $\widehat{B}^*_r$ is constant on its $V_1 \times V_2$ block, and each entry on the block is given by
\begin{align}
\ip{D}{\frac{v_1 v_2^T}{\norm{D}{v_1 v_2^T}^2}}{D^{-1}A^* D^{-1}} = \frac{v_2^T A^* v_1}{\left(v_1^T D v_1\right)\left(v_2^T D v_2\right)} = O\left(\frac{w^*(V_1, V_2)}{\vol(V_1) \vol(V_2)}\right). \label{eq:qurentries}
\end{align}

Using this, we bound
\begin{align*}
\norm{2}{\phi^{(r)}}^2 = \sum \limits_{i=1}^r \norm{D}{A^*_i}^2 = \norm{D}{\sum \limits_{i=1}^r A^*_i}^2 &= \norm{D}{\widehat{A}^*_r}^2 \tag*{$\{A^*_i\}$ orthogonal}\\
&\le \norm{D}{\widehat{B}^*_r}^2\\
&= \norm{F}{D^{1/2} \widehat{B}^*_r D^{1/2}}^2\\
&= \sum \limits_{V_1, V_2 \in \Pi} \sum \limits_{a \in V_1, b \in V_2} \left(\widehat{B}^*_r\right)^2_{ab} D_{aa} D_{bb} \\
&= \sum \limits_{V_1, V_2 \in \Pi} \vol(V_1) \vol(V_2) \cdot O\left(\frac{w^*(V_1, V_2)}{\vol(V_1) \vol(V_2)}\right)^2 \tag*{using (\ref{eq:qurentries})}\\
&= \vol(V)^2 \cdot  \sum \limits_{V_1, V_2 \in \Pi} \frac{\vol(V_1) \vol(V_2)}{\vol(V)^2} \cdot O\left(\frac{w^*(V_1, V_2)}{\vol(V_1)\vol(V_2)}\right)^2\\
&= \vol(V)^2 \cdot O\left(\frac{ w^*(E) }{\vol(V)^2} \right)^2 \tag*{$G$ is $L^2$ dense, $|\Pi| \le \eta^{-1}$}\\
&= O\left( \frac{w^*(E)}{\vol(V)} \right)^2\\
&= O\left( \frac{w^*(E)}{\onevec^T D \onevec} \right)^2. \tag*{\qedhere}
\end{align*}
\end{proof}

\subsection{Core Density and Low Threshold Rank \label{sec:ltr}}

In this part, let $\deg(v)$ denote the \emph{weighted} degree of the node $v$; that is, $\deg(v) := \sum_{v \in e} w(e)$.

\begin{definition} [Core Density \cite{KV09}] \label{def:coredense}
For a graph $G = (V, E)$ with unit node weights, the \emph{core density} is the quantity
$$\sum \limits_{(u, v) \in E} \frac{1}{\left(deg(u) + \overline{d}\right)\left(deg(v) + \overline{d}\right)}.$$
We say that $G$ is \emph{core dense} if its core density is $O(1)$.
\end{definition}

\begin{theorem} \label{thm:coredensepr}
Let $G$ be a core dense graph with average weighted node degree $\overline{d}$.
Then $G$ is $r$ cut pseudorandom, for any parameter $r$.
\end{theorem}
\begin{proof}
Let $A$ be the (edge- or fully-) weighted adjacency matrix of $G$, and let $D$ be the diagonal matrix where $D_{ii} := \deg(v_i) + \overline{d}$.
The entry $(u, v)$ of the matrix $D^{-1/2} A D^{-1/2}$ is exactly
$$\frac{1}{\sqrt{\left(deg(u) + \overline{d}\right)\left(deg(v) + \overline{d}\right)}},$$
so the sum square of entries of this matrix is exactly the core density of $A$.
Thus, since $G$ is core dense, we have
$$\norm{D^{-1}}{A}^2 = \norm{F}{D^{-1/2} A D^{-1/2}}^2 = O(1).$$
Using this together with Lemma \ref{lem:dinvbound}, letting $\Phi = \phi^{(n^2)}$ be the vector holding \emph{all} projection values of the decomposition of $G$ with respect to $D$, we have
\begin{align*}
\norm{2}{\Phi}^2 \le \norm{F}{D^{-1/2} A D^{-1/2}}^2 = O(1).
\end{align*}
Additionally, in $G$ we have
$$\onevec^T D \onevec = \sum \limits_{v \in V} \deg(v) + \overline{d} = 2w(E) + 2w(E) = O(w(E)),$$
and so
\begin{align*}
\frac{w(E)}{\onevec^T D \onevec} = O(1),
\end{align*}
thus implying cut pseudorandomness.
\end{proof}

Finally, we consider graphs of \emph{low threshold rank}:

\begin{definition} [Low Threshold Rank Graphs \cite{BRS11, GS11, GS12, GS13, GharanT15}]
Let $G$ be a graph with unit node weights and (edge- or fully-)weighted adjacency matrix $A$, and let $D$ be the diagonal matrix where $D_{ii} = \deg(v_i)$.
We say that $G$ has \emph{low $\eps$ threshold rank} if
$$\sum \limits_{\lambda \text{ is an eigenvalue of } D^{-1/2}AD^{-1/2} , \lambda > \eps} \lambda^2 = O\left( 1 \right).$$
\end{definition}

\begin{theorem}
For any $\eps > 0$, any graph $G$ with low $\eps$ threshold rank is $r = \Theta(\eps^{-2})$ cut pseudorandom.
\end{theorem}
\begin{proof}
Let $\Lambda^{(r)}$ be the vector containing the leading $r$ eigenvalues of the matrix $D^{-1/2} A D^{-1/2}$.
Using Lemma \ref{lem:dinvbound}, when we decompose $A$ with respect to $D$, the projection values $\{\phi_i\}$ satisfy
\begin{align*}
\norm{2}{\phi^{(r)}}^2 &\le \norm{2}{\Lambda^{(r)}}^2\\
&= \left(\sum \limits_{\lambda \in \Lambda^{(r)}, \lambda > \eps} \lambda^2 \right) + \left(\sum \limits_{\lambda \in \Lambda^{(r)}, \lambda \le \eps } \lambda^2 \right)\\
&\le O\left(1\right) + r \eps^2\\
&= O(1)
\end{align*}
where the second-to-last inequality is obtained using the definition of cut pseudorandomness and a union over the $r$ eigenvalues in $\Lambda^{(r)}$, and the last one is obtained using the parameter restriction $r = O(\eps^{-2})$.
Finally, as before we notice that $\onevec^T D \onevec = 2w(E)$, so $w(E) / (\onevec^T D \onevec) = O(1)$, thus implying cut psuedorandomness.
\end{proof}

\section{Algorithmic Applications \label{sec:algorithms}}

Here we will survey some of the algorithmic applications that arise from our new regularity lemmas.

\subsection{Graph Algorithms}

Oveis Gharan and Trevisan \cite{GharanT15} developed the following algorithmic applications of cut decompositions:
\begin{theorem} [\cite{GharanT15}] \label{thm:algs}
Let $\eps > 0$, let $G = (V, E)$ be a node- and edge-\textbf{un}weighted graph, and let $H$ be a graph on the same vertex set $V$ which may have node and edge weights.
Suppose $H$ satisfies the premises of Lemma \ref{lem:sparsecut} -- that is,
\begin{itemize}
\item The edge-weighted adjacency matrix of $H$ is the sum of $O(\eps^{-2})$ cut matrices (given as input), and for each cut matrix $C$ we have $\left|w^*(C)\right| = O(w^*(E))$, and
\item $\norm{\square}{A_G- A^*_H} = O(\eps |E|)$, where $A_G$ is the adjacency matrix of $G$ and $A^*_H$ is the fully-weighted adjacency matrix of $H$.
\end{itemize}
Then in $2^{\widetilde{O}\left(1/\eps^3\right)} + \poly(n)$ time, one can find any of:
\begin{itemize}
\item A MAX-CUT on $G$ within $\pm \eps |E|$ error,
\item A MAX-BISECTION on $G$ within $\pm \eps |E|$ error, and
\item A MIN-BISECTION on $G$ within $\pm \eps |E|$ error.  
\end{itemize}
\end{theorem}
Note that MAX-CUT, MAX-BISECTION, and MIN-BISECTION all have $\text{OPT} = \Theta(|E|)$, so this algorithm is an efficient FPTAS.
Theorem \ref{thm:algs} is highly nontrivial; the proof includes a number of detailed and clever technical ideas that we will not recap here.
Their proof is for a specific node-weight function for $H$ (specifically, each node $v$ is weighted by $\deg(v)^{-1/2}$), but it extends easily to arbitrary positive node weights.

In \cite{GharanT15}, the authors observe that graphs of low threshold rank satisfy the premises, and so they enjoy these approximation algorithms.
By Lemma \ref{lem:sparsecut}, the premises of Theorem \ref{thm:algs} also hold for the appropriate class of cut pseudorandom graphs.
We thus have: 
\begin{corollary}
For any $\eps > 0$ and any unweighted graph that is $r = O(\eps^{-2})$ cut pseudorandom (with respect to any matrix $D$), one can compute MAX-CUT, MAX-BISECTION, and MIN-BISECTION with an efficient FPTAS.
\end{corollary}

(The hypothesis that $\left|w^*(C)\right| = O(w^*(E))$ was not stated explicitly in the proof of Lemma \ref{lem:sparsecut}, but it follows instantly from the proof.)


\subsection{MAX CSP Algorithms and Tensor Decompositions}

MAX-CUT can be viewed as MAX-$2$-CSP; this naturally leads to the question of whether the regularity method also gives PTASes for MAX-$k$-CSP.
Such a PTAS was proved by Kannan and Vempala in \cite{KV09}, for core-dense graphs.
The core of their approach was the following tensor decomposition theorem.
Recall that a tensor $T \in \rr^{n_1 \times n_2 \times \ldots n_s}$ is an $s$-dimensional array.
While there is no direct analog of SVD for tensors, we observe next that the cut decomposition extends naturally.
For a tensor $T$ as above, define its $s$-form 
\[
T(x^1, \ldots, x^s) = \sum_{i_1,i_2,\ldots, i_s = 1}^{n_1, \ldots, n_s}T_{i_1,\ldots,i_r}x^1_{i_1}\ldots x^r_{i_s},
\]
and define the cut norm of $T$ as 
\[
\norm{\square}{T} :=  \max \limits_{(u^1,\ldots,u^s) \in X} \left| T(u^1,\ldots,u^s) \right|.
\]
\begin{theorem}\label{thm:tensor-pdecomp}
Let $T \in \rr^{n_1 \times \ldots \times n_s}$ and let $\eps > 0$.
Then there is $\widehat{T}$ that is a linear combination of $r = O(\eps^{-2})$ ``cut tensors,'' i.e., each is formed by an outer product of a tuple of cut vectors, satisfying
$$\norm{\square}{T - \widehat{T}} = O\left( \eps \norm{F}{T} \sqrt{n_1 \times \dots \times n_s} \right).$$
\end{theorem}

Here $\norm{F}{T}$ denotes the $L^2$ vector norm of the entries of $T$.
We will point out that our cut decomposition approach can be extended to prove this tensor decomposition theorem \emph{existentially}.
This is only half the work needed for a PTAS: one also must compute the decomposition quickly, and unfortunately, our algorithms do not readily extend to tensors.
However, Kannan and Vempala \cite{KV09} gave a different algorithm that computed the decomposition in the case of core-dense tensors (we refer to \cite{KV09} for details on how the definition of core-dense matrices extends to tensors).

Define the inner product $\ip{F}{\cdot}{\cdot}$ over tensors that works as the standard Euclidean inner product over their entries.
Let $S$ be an initially-empty set of tensors, and at all times we define the tensor $\widehat{T}$ as the projection of the input tensor $T$ onto the span of the tensors in $S$, under $\ip{F}{\cdot}{\cdot}$.
In each round, we find the cut vectors maximizing the quantity
$$\arg \max \limits_{(u^1,\ldots,u^s) \in X} \left| \frac{(T - \widehat{T})(u^1,\ldots,u^s)}{\|u^1\|_2, \dots, \|u^s\|_2} \right|.$$
Then, we add the tensor $u^1 \otimes \dots \otimes u^s$ to $S$ (where $\otimes$ denotes outer product).
Letting $\widehat{T}_i$ be the value of $\widehat{T}$ in the $i^{th}$ round of the decomposition, and $T_i := \widehat{T}_i - \widehat{T}_{i-1}$, we define the $i^{th}$ projection value as $\phi_i := \norm{F}{T_i}$.
Similar to the case of matrices, the tensors $\{T_i\}$ are orthogonal under $\ip{F}{\cdot}{\cdot}$, and it follows that $\phi_i \le \eps \norm{F}{T - \widehat{T}_i}$ for some $i = O(\eps^{-2})$.
This matrix $\widehat{T}_i$ satisfies the theorem, and the proof is completed by the chain of inequalities
$$\frac{\cutnorm{T - \widehat{T}_i}}{\sqrt{n_1 \times \dots \times n_s}} \le \max \limits_{(u^1,\ldots,u^s) \in X} \left| \frac{(T - \widehat{T}_i)(u^1,\ldots,u^s)}{\|u^1\|_2, \dots, \|u^s\|_2} \right| \le \phi_i \le \eps \norm{F}{T - \widehat{T}_i}.$$

%


\section*{Acknowledgments}

We are grateful to Mina Dalirrooyfard, Keaton Hamm, and to several anonymous reviewers for corrections, useful technical discussions, and references to prior work.

\bibliographystyle{acm}
\bibliography{regrefs}

\appendix

\section{Cut Norm Maximization \label{app:cutpvdalg}}

Here, we prove that our cut decompositions are computable in polynomial time with respect to any diagonal matrix $D$ with positive polynomially-bounded integer entries on the diagonal.
Recall that the goal is to take an arbitrary matrix\footnote{In our applications, we apply this algorithm to matrices of the form $A - D \widehat{A}_i D$, but by reparametrizing $A \gets A - D \widehat{A}_i D$ it suffices to focus on arbitrary matrices $A$.} $A$ and a diagonal matrix $D$, and compute

\begin{align*}
\max \limits_{v, w \text{ cut vectors}} \left|\ip{D}{\frac{vw^T}{\norm{D}{vw^T}}}{D^{-1} A D^{-1}}\right| &= \max \limits_{v, w \text{ cut vectors}} \left|\frac{w^T A v}{\norm{2}{w^T D w} \norm{2}{v^T D v}}\right|\\
&= \max \limits_{S, T \subseteq [n]} \left|\frac{A(S, T)}{\sqrt{\sum \limits_{i \in S} d_i \sum \limits_{j \in T} d_j}}\right|
\end{align*}
where $d_i := D_{ii}$.
We may drop the absolute value symbols in the maximized expression by running our algorithm twice, once on $A$ and once on $-A$, and taking the better of the two solutions.
So we will now show the following:

\begin{theorem} \label{thm:algorithmiccutpvd}
There is an algorithm that, given $A \in \rr^{n \times n}$ 
and nonnegative vector of integers $d\in {\cal Z}^{n}$, returns $S^*, T^* \subseteq [n]$ exactly solving 
\[
\max_{S, T \subseteq [n]} \frac{A(S, T)}{\sqrt{\sum \limits_{i\in S} d_i \sum \limits_{j\in T} d_j}}
\]
in time polynomial in $n$, the description length of $A$ (number of bits) and $\sum_{i \in [n]} d_i$.
\end{theorem}

Our algorithm is closely based on that of Charikar \cite{Charikar00}, which solves the special case where $A$ is a binary matrix.
Charikar's algorithm extends almost immediately to handle the presence of the node weights $\{d_i\}$, and to handle the presence of arbitrary nonnegative entries of the input matrix $A$, but we introduce some nontrivial additional machinery (essentially, the last constraint in the following LP) in order to handle the possibility of negative entries in the input matrix.

The first step in the algorithm is to guess the values
$$c_S := \sum_{i\in S}d_i \quad \text{and} \quad c_T := \sum_{j\in T} d_j,$$
noting that the number of possible values of $(c_S, c_T)$ is at most
$$\left(\sum_{i \in [n]} d_i\right)^2,$$
and so we pay this factor in runtime to iterate over all possible guesses for $c_S, c_T$.
We define
$$OPT[c_S, c_T] := \left(\max_{\substack{S, T \subseteq [n]\\\sum_{i \in S} d_i = c_S\\\sum_{j \in T} d_j = c_T}} \frac{A(S, T)}{\sqrt{\sum \limits_{i\in S} d_i \sum \limits_{j\in T} d_j}}\right)$$
as the maximum value of the expression in Theorem \ref{thm:algorithmiccutpvd}, over subsets $S, T$ that agree with the guesses $c_S, c_T$.
The second step in the algorithm is to set
$$A' := A + \lambda dd^T,$$
where $d \in \rr^n$ is the diagonal of $D$, and $\lambda \ge 0$ is a scalar chosen sufficiently large so that all entries of $A'$ are positive.
Next, consider the following linear program (which depends on our guesses of $c_S, c_T$):
\begin{mdframed}
\begin{align*}
\text{choose } & \{s_i\}_{i \in [n]}, \{t_j\}_{j \in [n]}, \{x_{ij}\}_{i,j \in [n]}\\
\text{to maximize } & \sum_{i,j\in [n]} A'_{ij} x_{ij} \\
\text{such that:} &\\
& 0 \le x_{ij} \le s_i \text{ and } 0 \le x_{ij} \le t_j && \text{for all } i,j\\
& 0 \le s_i \le \frac{1}{\sqrt{c_S c_T}} \text{ and } 0 \le t_j \le \frac{1}{\sqrt{c_S c_T}} && \text{for all } i,j\\
& \sum_{i \in [n]} d_i s_i = \sqrt{\frac{c_S}{c_T}} \\
& \sum_{j\in [n]} d_j t_j = \sqrt{\frac{c_T}{c_S}}\\
& \sum_{i, j \in [n]} d_i d_j x_{ij} = \sqrt{c_S c_T}.
\end{align*}
\end{mdframed}
We will write $OPT'[c_S, c_T]$ for the maximizing value of this linear program.
We first show that this LP has a solution whose quality depends on the maximizing value:
\begin{claim} \label{clm:optub}
For each $c_s, c_T$, we have
$$OPT'[c_S, c_T] \ge OPT[c_S, c_T] + \lambda \sqrt{c_S c_T}.$$
\end{claim}
\begin{proof}
For any two subsets $S, T$ that agree with our guesses for $c_S, c_T$, consider the variable setting
$$s_i := \begin{cases}
\frac{1}{\sqrt{c_S c_T}} & \text{if } i \in S\\
0 & \text{if } i \notin S
\end{cases}$$
$$t_j := \begin{cases}
\frac{1}{\sqrt{c_S c_T}} & \text{if } j \in T\\
0 & \text{if } j \notin T
\end{cases}$$
$$x_{ij} := \begin{cases}
\frac{1}{\sqrt{c_S c_T}} & \text{if } i \in S \text{ and } j \in T\\
0 & \text{otherwise.}
\end{cases}$$
Let us show that this variable setting is feasibile.
The inequalities in the first two lines follow immediately from the variable choices.
For the equality in the third line, we have 
\begin{align*}
\sum \limits_{i \in [n]} d_i s_i &= \frac{1}{\sqrt{c_S c_T}} \cdot \sum \limits_{i \in S} d_i\\
&= \frac{1}{\sqrt{c_S c_T}} \cdot c_S\\
&= \sqrt{\frac{c_S}{c_T}},
\end{align*}
and a similar calculation holds to verify the fourth line equality on $\sum_{j \in T} d_j t_j$.
For the fifth line, we have
\begin{align*}
\sum \limits_{i,j \in [n]} d_i d_j x_{ij} &= \sqrt{c_S c_T} \cdot \left( \sum \limits_{i \in S} d_i s_i \right) \left( \sum \limits_{j \in T} d_j t_j\right)\\
&= \sqrt{c_S c_T} \cdot \sqrt{\frac{c_S}{c_T}} \cdot \sqrt{\frac{c_T}{c_S}} \tag*{by previous calculations}\\
&= \sqrt{c_S c_T}.
\end{align*}
Now we compute the value of the variable setting.
For pairs of indices $i,j$ with $i \notin S$ or $j \notin T$, we set $x_{ij}=0$ and so this pair of indices does not affect the objective.
For pairs of indices $i,j$ with $i \in S$ and $j \notin T$, we set
$$x_{ij} = \frac{1}{\sqrt{c_S c_T}},$$
and so the total contribution to the objective over all such indices is
\begin{align*}
\sum \limits_{i \in S, j \in T} \frac{A'_{ij}}{\sqrt{c_S c_T}} &= \frac{A'(S, T)}{\sqrt{c_S c_T}}\\
&= \frac{A(S, T) + \sum \limits_{i \in S, j \in T} \lambda d_i d_j}{\sqrt{c_S c_T}}\\
&= \frac{A(S, T) + \lambda\left(\sum \limits_{i \in S} d_i\right)\left(\sum \limits_{j \in T} d_j\right)}{\sqrt{c_S c_T}}\\
&= \frac{A(S, T) + \lambda c_S c_T}{\sqrt{c_S c_T}}\\
&= \frac{A(S, T)}{\sqrt{c_S c_T}} + \lambda \sqrt{c_S c_T},
\end{align*}
as desired.
\end{proof}

Next we want to show that, given an arbitrary maximizing variable setting for this LP, one can extract explicit subsets $(S^*, T^*)$ realizing this value.
We have:
\begin{claim} \label{clm:optlb}
Given an optimizing solution to the LP $\{x_{ij}\}, \{s_i\}, \{t_j\}$, in polynomial time we can find subsets $S^*, T^*$ with
$$\frac{A(S^*, T^*)}{\sqrt{\sum \limits_{i \in S^*} d_i \sum \limits_{j \in T^*} d_j}} + \lambda \sqrt{c_S c_T} \ge OPT'[c_S, c_T].$$
(These subsets $S^*, T^*$ might not agree with our choice of $c_S, c_T$.)
\end{claim}
\begin{proof}
For any parameter $r \ge 0$, define
$$S(r) := \{i \in [n] \ \mid \ s_i \ge r\}, \quad T(r) := \{j \in [n] \ \mid \ t_j \ge r\}, \quad E(r) := \{(i, j) \ \mid \ x_{ij} \ge r\}.$$
Note that since each $A'_{ij}$ is positive, in an optimal variable setting we will set each $x_{ij} = \min\{s_i, t_j\}$ and therefore $E(r) = S(r) \times T(r)$.
It will be useful to observe a control on the following three integral expressions:
$$\int \limits_{0}^{\infty} \left( \sum \limits_{i \in S(r)} d_i\right) \ dr = \sum \limits_{i \in [n]} d_i s_i = \sqrt{\frac{c_S}{c_T}},$$
$$\int \limits_{0}^{\infty} \left( \sum \limits_{j \in T(r)} d_j\right) \ dr = \sum \limits_{j \in [n]} d_j t_j = \sqrt{\frac{c_T}{c_S}},$$
$$\int \limits_{0}^{\infty} \left( \sum \limits_{ij \in E(r)} A_{ij} \right) \ dr = \sum \limits_{i, j \in [n]} A_{ij} x_{ij}.$$
Thus,
\begin{align*}
OPT'[c_S, c_T] &= \sum \limits_{i, j \in [n]} A'_{ij} x_{ij}\\
&= \left(\sum \limits_{i, j \in [n]} A_{ij} x_{ij} \right) + \left(\lambda \sum \limits_{i, j \in [n]} d_i d_j x_{ij}\right)\\
&= \left(\sum \limits_{i, j \in [n]} A_{ij} x_{ij} \right) + \lambda \sqrt{c_S c_T} \tag*{last constraint in LP}\\
&= \int \limits_{0}^{\infty} \left( \sum \limits_{ij \in E(r)} A_{ij}\right) \ dr + \lambda \sqrt{c_S c_T}\\
&= \frac{\int \limits_{0}^{\infty} \left( \sum \limits_{ij \in E(r)} A_{ij}\right) \ dr}{\left(\sqrt{\frac{c_S}{c_T}} \cdot \sqrt{\frac{c_T}{c_S}}\right)^{1/2}} + \lambda \sqrt{c_S c_T}\\
&= \frac{\int \limits_{0}^{\infty} \left( \sum \limits_{ij \in E(r)} A_{ij}\right) \ dr}{\left(\int \limits_{0}^{\infty} \left( \sum \limits_{i \in S(r)} d_i\right) \ dr \cdot \int \limits_{0}^{\infty} \left( \sum \limits_{j \in T(r)} d_j\right) \ dr\right)^{1/2}} + \lambda \sqrt{c_S c_T}\\
&\le \frac{\int \limits_{0}^{\infty} \left( \sum \limits_{ij \in E(r)} A_{ij}\right) \ dr}{\int \limits_{0}^{\infty} \sqrt{\sum \limits_{i \in S(r)} d_i \sum \limits_{j \in T(r)} d_j}\ dr} + \lambda \sqrt{c_S c_T} \tag*{Cauchy-Schwarz inequality.}\\
\end{align*}

Thus, there must exist a particular choice of $r$ for which
\begin{align*}
OPT'[c_S, c_T] &\le \frac{\sum \limits_{ij \in E(r)} A_{ij}}{\sqrt{\sum \limits_{i \in S(r)} d_i \sum \limits_{j \in T(r)} d_j}} + \lambda \sqrt{c_S c_T}\\
&= \frac{A(S(r), T(r))}{\sqrt{\sum \limits_{i \in S(r)} d_i \sum \limits_{j \in T(r)} d_j}} + \lambda \sqrt{c_S c_T},
\end{align*}
showing that the sets $S(r), T(r)$ satisfy the lemma.
We can find the appropriate choice of $r$ in polynomial time by scanning, since we only need to check values of $r$ at which nodes leave $S(r), T(r)$, and there are only $O(n)$ many such thresholds.
\end{proof}

We now wrap up the proof.
In every round of our algorithm, by the previous two claims, we compute subsets $S^*, T^*$ satisfying the inequality
$$\frac{A(S^*, T^*)}{\sqrt{\sum \limits_{i \in S^*} d_i \sum \limits_{j \in T^*} d_j}} + \lambda \sqrt{c_S c_T} \ge OPT'[c_S, c_T] \ge OPT[c_S, c_T] + \lambda \sqrt{c_S c_T}$$
and so
$$\frac{A(S^*, T^*)}{\sqrt{\sum \limits_{i \in S^*} d_i \sum \limits_{j \in T^*} d_j}} \ge OPT[c_S, c_T].$$
Thus, in at least one round of the algorithm, our subsets $S^*, T^*$ satisfy
$$\frac{A(S^*, T^*)}{\sqrt{\sum \limits_{i \in S^*} d_i \sum \limits_{j \in T^*} d_j}} \ge \max \limits_{c_S, c_T} OPT[c_S, c_T].$$
By definition the right-hand side of this inequality is precisely the maximum of the left-hand side over all possible choices of $S^*, T^*$.
Thus, this particular choice of subsets $S^*, T^*$ maximizes the expression in Theorem \ref{thm:algorithmiccutpvd}, as desired.

\end{document}